\documentclass{article}
\usepackage{amsthm}
\usepackage{xspace}
\usepackage{fullpage}
\usepackage[utf8]{inputenc}
\usepackage{graphicx}
\usepackage{amsmath}
\usepackage{float}
\usepackage{amsfonts}
\usepackage[ruled,noend,linesnumbered]{algorithm2e}     
\usepackage{enumerate,comment}
\usepackage{url}

\DeclareMathOperator{\per}{per}
\newcommand{\integ}{\mathbb{Z}}

\newcommand{\eps}{\varepsilon}
\newcommand{\ceil}[1]{\lceil #1 \rceil}
\newcommand{\floor}[1]{\lfloor #1 \rfloor}
\newcommand{\llist}{\mathcal{L}} 
\newcommand{\occ}{\mathsf{occ}}
\newcommand{\poly}{\text{poly}}

\newcommand{\Oh}{\mathcal{O}}

   \newtheorem{theorem}{Theorem}
   
   \newtheorem{lemma}[theorem]{Lemma}
   \newtheorem{observation}[theorem]{Observation}
   \newtheorem{fact}[theorem]{Fact}
   \theoremstyle{definition}   
   \newtheorem{definition}[theorem]{Definition}
   \usepackage{authblk}
   \theoremstyle{remark}

\usepackage{thmtools}
\usepackage{thm-restate}
\usepackage{todonotes}

\begin{document}

\title{Approximating LZ77 via\\Small-Space Multiple-Pattern Matching}

\author[1]{Johannes Fischer\thanks{Supported by Academy of Finland grant 268324.}}
\author[2]{Travis Gagie}
\author[3]{Paweł Gawrychowski\thanks{Work done while the author held a post-doctoral position at Warsaw Center of Mathematics and Computer Science.}}
\author[3]{Tomasz Kociumaka\thanks{Supported by Polish budget funds for science in 2013-2017 as a research project under the `Diamond Grant' program.}}

\affil[1]{TU Dortmund, Germany}
\affil[ ]{\texttt{johannes.fischer@cs.tu-dortmund.de}}
\affil[2]{Helsinki Institute for Information Technology (HIIT), 
Department of Computer Science, University of Helsinki, Finland}
\affil[ ]{\texttt{travis.gagie@cs.helsinki.fi}}
\affil[3]{Institute of Informatics,
    University of Warsaw, Poland}
\affil[ ]{\texttt{\{gawry,kociumaka\}@mimuw.edu.pl}}

\maketitle

\begin{abstract}
We generalize Karp-Rabin string matching to handle multiple patterns in $\Oh(n \log n + m)$ time and $\Oh(s)$ space, where
$n$ is the length of the text and $m$ is the total length of the $s$ patterns, returning correct answers with high probability.
As a prime application of our algorithm, we show how to approximate the LZ77 parse of a string of length $n$.
If the optimal parse consists of $z$ phrases, using only $\Oh(z)$ working space we can return a parse consisting of at most $(1+\eps)z$ phrases in $\Oh(\eps^{-1}n\log n)$ time, for any $\eps\in (0,1]$.
As previous quasilinear-time algorithms for LZ77 use $\Omega(n/\poly\log n)$ space, but $z$ can be exponentially small in $n$, these improvements in space are substantial.
\end{abstract}

\section{Introduction}
\label{sec:introduction}

Multiple-pattern matching, the task of locating the occurrences of $s$ patterns of total length $m$ in a single text of length $n$,
is a fundamental problem in the field of string algorithms. 
The algorithm by Aho and Corasick~\cite{AC} solves this problem using $\Oh (n + m)$ time and $\Oh (m)$ working space in addition to the space needed for the text and patterns.
To list all $\occ$ occurrences rather than, e.g., the leftmost ones, extra $\Oh(\occ)$ time is necessary.
When the space is limited, we can use a compressed
Aho-Corasick automaton~\cite{HonKSTV13}. In extreme cases, one could apply a linear-time constant-space single-pattern
matching algorithm sequentially for each pattern in turn, at the cost of increasing the running time to $\Oh (n\cdot s+m)$. Well-known examples of such algorithms include
those by Galil and Seiferas~\cite{GS83}, Crochemore and Perrin~\cite{CP91}, and Karp and Rabin~\cite{KR87} (see~\cite{BreslauerGM13} for a recent survey).

It is easy to generalize Karp-Rabin matching to handle multiple patterns in $\Oh(n + m)$ expected time and $\Oh(s)$ working space provided that
all patterns are of the same length~\cite{GumLipton}. 
To do this, we store the fingerprints of the patterns in a hash table, and then slide a window over
the text maintaining the fingerprint of the fragment currently in the window.
The hash table lets us check if the fragment is an occurrence of a pattern.
If so, we report it and update the hash table so that every pattern is returned at most once.
This is a very simple and actually applied idea~\cite{Wikipedia}, but it is not clear how to extend it for patterns with many distinct lengths.
In this paper we develop a dictionary matching algorithm which works for any set of patterns in $\Oh(n\log n +m)$ time and $\Oh(s)$ working space,
assuming that read-only random access to the text and the patterns is available.
If required, we can compute for every pattern its longest
prefix occurring in the text, also in $\Oh(n\log n+m)$ time and $\Oh(s)$ working space.

In a very recent independent work Clifford et al.~\cite{esa2} gave a dictionary matching algorithm in the streaming model.
In this setting the patterns and later the text are scanned once only (as opposed to read-only random access) and an occurrence
needs to be reported immediately after its last character is read. 
Their algorithm uses $\Oh(s\log \ell)$ space
and takes $\Oh(\log \log (s+\ell))$ time per character where $\ell$ is the length of the longest pattern ($\frac{m}{s}\le \ell \le m$). Even though some of the ideas used in both results are similar,
one should note that the streaming and read-only models are quite different.
In particular, computing the longest prefix occurring in the text for every
pattern requires $\Omega(m\log \min(n,|\Sigma|))$  bits of space in the streaming model, as opposed to the $\Oh(s)$ working
space achieved by our solution in the read-only setting.

As a prime application of our dictionary matching algorithm, we show how to approximate the Lempel-Ziv 77 (LZ77) parse~\cite{ziv77universal} of a text of length $n$ using working space proportional to the number of phrases (again, we assume read-only random access to the text).
Computing the LZ77 parse in small space is an issue of high importance, with space being a frequent bottleneck of today's systems.
Moreover, LZ77 is useful not only for data compression, but also as a way to speed up algorithms~\cite{LohreySurvey}.
We present a general approximation algorithm working in $\Oh(z)$ space for inputs
admitting LZ77 parsing with $z$ phrases. For any $\eps \in (0,1]$, the algorithm can be used
to produce a parse consisting of $(1+\eps)z$ phrases in $\Oh(\eps^{-1}n\log n)$ time.

To the best of our knowledge, approximating LZ77 factorization in small space has not been considered before,
and our algorithm is significantly more efficient than methods producing the exact answer.
A recent sublinear-space algorithm, due to Kärkkäinen et al.~\cite{kaerkkaeinen13lightweight}, runs in $\Oh(nd)$ time and uses $\Oh(n/d)$ space, for any parameter $d$. 
An earlier online solution by Gasieniec et al.~\cite{DBLP:conf/swat/GasieniecKPR96}
uses $\Oh(z)$ space and takes $\Oh(z^2\log^2z)$ time for each character appended.
Other previous methods use significantly more space when the parse is small relative to $n$; see~\cite{2015arXiv150402605F}
for a recent discussion.

\paragraph{\bfseries{Structure of the paper.}}
Sect.~\ref{sec:preliminaries} introduces terminology and recalls several known concepts. 
This is followed by the description of our dictionary matching algorithm.
In Sect.~\ref{sec:short} we show how to process patterns of length at most $s$ and
in Sect.~\ref{sec:long} we handle longer patterns, with different procedures for repetitive and non-repetitive ones.
In Sect.~\ref{sec:prefixes} we extend the algorithm to compute, for every pattern,
the longest prefix occurring in the text.
Finally, in Sect.~\ref{sec:lz77approx}, we apply the dictionary matching algorithm to construct an approximation of the LZ77 parsing, and in Sect.~\ref{app:lv} we explain how to modify the algorithms to make them Las Vegas.

\paragraph{\bfseries{Model of computation.}}
Our algorithms are designed for the word-RAM with $\Omega (\log n)$-bit words and assume integer alphabet of polynomial size.
The usage of Karp-Rabin fingerprints makes them Monte Carlo randomized:
the correct answer is returned with high probability, i.e., the error probability is inverse polynomial with respect to input size,
where the degree of the polynomial can be set arbitrarily large. 
With some additional effort, our algorithms can be turned into Las Vegas randomized, where the answer is always correct and the time bounds hold with high probability.
Throughout the whole paper,
we assume read-only random access to the text and the patterns, and we do not
include their sizes while measuring space consumption.

\section{Preliminaries}
\label{sec:preliminaries}

We consider finite words over an integer alphabet $\Sigma=\{0,\ldots,\sigma-1\}$, where $\sigma=\poly(n+m)$.
For a word $w=w[1]\ldots w[n] \in \Sigma^n$, we define the \emph{length} of $w$ as $|w|=n$.
For $1\le i \le j \le n$, a word $u=w[i] \ldots w[j]$ is called a \emph{subword} of $w$.
By $w[i..j]$ we denote the occurrence of $u$
at position $i$, called a \emph{fragment} of~$w$.
A fragment with $i=1$ is called a \emph{prefix}
and a fragment with $j=n$ is called a \emph{suffix}.

A positive integer $p$ is called a period of $w$ whenever $w[i] = w[i+p]$ for all $i=1,2,\ldots,|w|-p$.
In this case, the prefix $w[1..p]$ is often also called a period of $w$.
The length of the shortest period of a word $w$ is denoted as $\per(w)$.
A word $w$ is called \emph{periodic} if $\per(w)\le \frac{|w|}{2}$
and \emph{highly periodic} if $\per(w)\le \frac{|w|}{3}$.
The well-known periodicity lemma~\cite{fine1965uniqueness} says
that if $p$ and $q$ are both periods of $w$, and $p+q\leq |w|$, then $\gcd(p,q)$ is also
a period of $w$.
We say that word $w$ is \emph{primitive} if $\per(w)$ is not a proper divisor of $|w|$. Note that the shortest period
$w[1..\per(w)]$ is always primitive.

\subsection{Fingerprints}
\label{sec:fingerprints}

Our randomized construction is based on Karp-Rabin fingerprints; see \cite{KR87}.
Fix a word $w[1..n]$
over an alphabet $\Sigma=\{0,\ldots,\sigma-1\}$, a constant $c \ge 1$, a prime number $p> \max(\sigma, n^{c+4})$, and
choose $x\in \integ_{p}$ uniformly at random. We define the fingerprint
of a subword $w[i..j]$ as $\Phi(w[i..j])=w[i] + w[i+1]x + \ldots +w[j]x^{j-i} \bmod{p}$.
With probability at least $1-\frac{1}{n^{c}}$, 
no two distinct subwords of the same length have equal fingerprints.
The situation when this happens for some two subwords
is called a \emph{false-positive}. From now on when stating the results we assume that there are no
false-positives to avoid repeating that the answers are correct with high probability. For dictionary matching,
we assume that no two distinct subwords of $w=TP_1\ldots P_s$ have equal fingerprints.
Fingerprints let us easily locate many patterns of the same length.
A straightforward solution described in the introduction builds a hash table mapping fingerprints to patterns.
However, then we can only guarantee that the hash table is constructed correctly with probability
$1-\Oh(\frac{1}{s^c})$ (for an arbitrary constant~$c$), and we would like to bound the error probability
by $\Oh(\frac{1}{(n+m)^c})$. Hence we replace hash table with a deterministic dictionary as explained
below.
Although it increases the time by $\Oh(s\log s)$, the extra term becomes absorbed
in the final complexities.

\begin{theorem}\label{thm:equal}
Given a text $T$ of length $n$ and patterns $P_1,\ldots,P_s$, each of length exactly $\ell$,
we can compute the
the leftmost occurrence of every pattern $P_i$ in $T$ using $\Oh(n+s\ell+s\log s)$ total time and $\Oh(s)$ space.
\end{theorem}

\begin{proof}
We calculate the fingerprint $\Phi(P_{j})$ of every pattern. Then we build in $\Oh(s\log s)$ time~\cite{DBLP:conf/icalp/Ruzic08}
a deterministic dictionary $\mathcal{D}$ with an entry mapping $\Phi(P_{j})$ to $j$.
For multiple identical patterns we create just one entry, and at the end we copy the answers to all instances of the pattern.
Then we scan the text $T$ with a sliding window of length $\ell$ while maintaining the fingerprint $\Phi(T[i..i+\ell-1])$ of the
current window. Using~$\mathcal{D}$, we can find in $\Oh(1)$ time an index $j$ such that $\Phi(T[i..i+\ell-1])=\Phi(P_{j})$, if any, and update
the answer for $P_{j}$ if needed (i.e., if there was no occurrence of $P_j$ before).
If we precompute $x^{-1}$, the fingerprints $\Phi(T[i..i+\ell-1])$ can be updated in $\Oh(1)$ time while increasing~$i$.

\end{proof}

\subsection{Tries}
A trie of a collection of strings $P_1,\ldots,P_s$ is a rooted tree
whose nodes correspond to prefixes of the strings. The root represents
the empty word and the edges are labeled with single characters. 
The node corresponding to a particular prefix is called its \emph{locus}.
In a \emph{compacted} trie unary nodes that do not represent any $P_i$
are \emph{dissolved} and the labels of their incidents edges are concatenated. 
The dissolved nodes are called \emph{implicit} as opposed to the \emph{explicit}
nodes, which remain stored. The locus of a string in a compacted trie might therefore be explicit or implicit.
All edges outgoing from the same node are stored on a list sorted according to the first character,
which is unique among these edges. 
The labels of edges of a compacted trie are stored as pointers to the respective fragments of strings $P_i$.
Consequently, a compacted trie can be stored in space proportional to the number of explicit nodes, which is $\Oh(s)$.

\newcommand{\T}{\mathcal{T}}
Consider two compacted tries $\T_1$ and $\T_2$. 
We say that (possibly implicit) nodes $v_1\in \T_1$ and $v_2\in \T_2$ are \emph{twins}
if they are loci of the same string. Note that every $v_{1}\in \T_{1}$ has at most one twin
$v_{2}\in \T_{2}$.

\begin{lemma}\label{lem:merge}
Given two compacted tries $\T_1$ and $\T_2$ constructed for $s_1$ and $s_2$ strings, respectively,
in $\Oh(s_{1}+s_{2})$ total time and space
we can find for each explicit node $v_1 \in \T_{1}$ a node $v_2\in \T_2$
such that if $v_1$ has a twin in $\T_2$, then $v_2$ is its twin.
(If $v_1$ has no twin in $\T_2$, the algorithm returns an arbitrary node $v_2\in \T_2$).
\end{lemma}

\begin{proof}
We recursively traverse both tries while maintaining a pair of nodes $v_1\in \T_1$ and $v_2\in \T_2$,
starting with the root of $\T_1$ and $\T_2$ satisfying the following invariant: either $v_{1}$ and $v_{2}$ are
twins, or $v_{1}$ has no twin in $\T_{2}$. If $v_{1}$ is explicit, we store $v_{2}$ as the
candidate for its twin.
Next, we list the (possibly implicit) children of $v_{1}$ and $v_{2}$ and 
match them according to the edge labels with a linear scan. We recurse on all pairs of matched children.
If both $v_1$ and $v_2$ are implicit, we simply advance to their immediate children.
The last step is repeated until we reach an explicit node in at least one of the tries,
so we keep it implicit in the implementation to make sure that the total number of operations is $\Oh(s_1+s_2)$.
If a node $v\in \T_{1}$ is not visited during the traversal, for sure it has no twin in $\T_{2}$.
Otherwise, we compute a single candidate for its twin.
\end{proof}

\section{Short Patterns}
\label{sec:short}

To handle the patterns of length not exceeding a given threshold $\ell$, we first build a compacted trie for those patterns. 
Construction is easy if the patterns are sorted lexicographically: we insert them one by one into the compacted trie first naively traversing the trie from the root, then potentially partitioning one edge into two parts, and finally adding a leaf if necessary. 
Thus, the following result suffices to efficiently build the tries.

\begin{lemma}\label{lem:sort}
One can lexicographically sort strings $P_1,\ldots,P_s$ of total length $m$ 
in $\Oh(m+\sigma^\eps)$ time using $\Oh(s)$ space, for any constant $\eps>0$.
\end{lemma}

\begin{proof}
We separately sort the $\sqrt{m}+\sigma^{\eps/2}$ longest strings and all the remaining strings, and then merge both sorted
lists. Note these longest strings can be found in $\Oh(s)$ time using a linear time selection algorithm.

Long strings are sorted using insertion sort. If the longest common prefixes between adjacent (in the sorted order) strings are 
computed and stored, inserting $P_j$ can be done in $\Oh(j+|P_j|)$ time. In more detail, let $S_1, S_2, \ldots, S_{j-1}$
be the sorted list of already processed strings. We start with $k:=1$ and keep increasing $k$ by one as long as $S_k$
is lexicographically smaller than $P_j$ while maintaining the longest common prefix between $S_k$ and $P_j$, denoted
$\ell$. After increasing $k$ by one, we update $\ell$ using the longest common prefix between $S_{k-1}$ and $S_k$,
denoted $\ell'$, as follows. If $\ell' > \ell$, we keep $\ell$ unchanged. If $\ell'=\ell$, we try to iteratively
increase $\ell$ by one as long as possible. In both cases, the new value of $\ell$ allows us to lexicographically
compare $S_k$ and $P_j$ in constant time.  Finally, $\ell' < \ell$ guarantees that $P_j<S_k$ and we may terminate the procedure.
Sorting the $\sqrt{m}+\sigma^{\eps/2}$ longest strings using this approach takes
$\Oh(m+(\sqrt{m}+\sigma^{\eps/2})^2)=\Oh(m+\sigma^{\eps})$ time.

The remaining strings are of length at most $\sqrt{m}$ each, and if there are any, then
$s\ge \sigma^{\eps/2}$. We sort these strings by iteratively applying radix sort, treating each symbol from
$\Sigma$ as a sequence of $\frac{2}{\eps}$ symbols from $\{0,1,\ldots,\sigma^{\eps/2}-1\}$.
Then a single radix sort takes time and space proportional to the number of strings involved plus the alphabet size,
which is $\Oh(s+\sigma^{\eps/2})=\Oh(s)$. Furthermore, because the numbers of strings involved in the subsequent
radix sorts sum up to $m$, the total time complexity is $\Oh(m+\sigma^{\eps/2}\sqrt{m})=\Oh(m+\sigma^{\eps})$.

Finally, the merging takes time linear in the sum of the lengths of all the involved strings, so the total complexity
is as claimed.
\end{proof}

Next, we partition $T$ into $\Oh(\frac{n}{\ell})$ overlapping blocks $T_1=T [1..2 \ell]$, $T_2= T [\ell+ 1..3 \ell]$, $T_3=T [2 \ell + 1..4 \ell], \ldots$.
 Notice that each subword of length at most $\ell$ is completely contained in some block.
Thus, we can consider every block separately.

The suffix tree of each block $T_i$ takes $\Oh(\ell\log \ell)$
time~\cite{Ukkonen95} and $\Oh(\ell)$ space to construct
and store (the suffix tree is discarded after processing the block).
We apply Lemma~\ref{lem:merge} to the suffix tree and the compacted trie of patterns; this takes $\Oh(\ell+s)$ time.
For each pattern $P_j$ we obtain a node such that the corresponding subword is equal to $P_j$
provided that $P_j$ occurs in $T_i$.
We compute the leftmost occurrence $T_i[b..e]$ of the subword, which takes constant time
if we store additional data at every explicit node of the suffix tree,
and then we check whether $T_i[b..e]=P_j$ using fingerprints.
For this, we precompute the fingerprints of all patterns, and for each block $T_i$ we precompute the fingerprints
of its prefixes in $\Oh(\ell)$ time and space, which allows to determine the fingerprint of any of its subwords in constant time.

In total, we spend $\Oh(m+\sigma^{\eps})$ for preprocessing and $\Oh(\ell\log \ell+s)$ for each block. 
Since $\sigma=(n+m)^{\Oh(1)}$, for small enough $\eps$ this yields the
following result.

\begin{theorem}\label{thm:short}
Given a text $T$ of length $n$ and patterns $P_1,\ldots,P_s$ of total length $m$,
using $\Oh(n\log\ell+s\frac{n}{\ell}+m)$ total time and $\Oh(s+\ell)$ space
we can compute
the leftmost occurrences in $T$ of every pattern $P_j$ of length at most $\ell$.
\end{theorem}

\section{Long Patterns}
\label{sec:long}

To handle patterns longer than a certain threshold, we first distribute them into groups according to the value of $\floor{\log_{4/3}|P_j|}$.
Patterns longer than the text can be ignored, so there are $\Oh(\log n)$ groups. 
Each group is handled separately, and from now on we consider only patterns $P_j$ satisfying $\floor{\log_{4/3}|P_j|}=i$.

We classify the patterns into classes depending on the periodicity of their prefixes and suffixes.  
We set $\ell = \ceil{(4 / 3)^{i}}$ and define $\alpha_j$ and $\beta_j$ as, respectively, the prefix and the suffix of length $\ell$ of $P_j$.  
Since $\frac{2}{3}(|\alpha_j|+|\beta_j|)=\frac{4}{3}\ell \ge |P_j|$, the following fact yields a classification of the patterns into three classes: either $P_j$ is highly periodic, or $\alpha_j$ is not highly periodic, or $\beta_j$ is not highly periodic.
The intuition behind this classification is that
if the prefix or the suffix is not repetitive, then we will not see it many times in a short subword of the text.
On the other hand, if both the prefix and suffix are repetitive, then there is some structure that we can take advantage of. 

\begin{fact}\label{fct:shr}
Suppose $x$ and $y$ are a prefix and a suffix of a word $w$, respectively.
If  $|x|+|y|\ge |w|+p$ and $p$ is a period of both $x$ and $y$, 
then $p$ is a period of~$w$.
\end{fact}

\begin{proof}
We need to prove that $w[i]=w[i+p]$ for all $i=1,2,\ldots,|w|-p$.
If $i+p\leq |x|$ this follows from $p$ being a period of $x$,
and if $i\geq |w|-|y|+1$ from $p$ being a period of $y$.
Because $|x|+|y| \ge |w|+p$, these two cases cover all possible
values of $i$.
\end{proof}

To assign every pattern to the appropriate class, we compute the periods of $P_j$, $\alpha_j$ and $\beta_j$
using small space. Roughly the same result has been proved in \cite{DBLP:conf/esa/KociumakaSV14},
but for completeness we provide the full proof here.

\begin{lemma}
\label{lem:comper}
Given a read-only string $w$ one can decide in $\Oh(|w|)$
time and constant space if $w$ is periodic and if so, compute $\per(w)$.
\end{lemma}

\begin{proof}
Let $v$ be the prefix of $w$ of length $\ceil{\frac12|w|}$ and $p$ be the starting position of the second occurrence of $v$ in $w$, if any.
We claim that if $\per(w) \leq \frac12|w|$, then $\per(w)=p-1$. Observe first that in this case $v$ occurs at a position $\per(w)+1$.
Hence, $\per(w)\ge p-1$. Moreover $p-1$ is a period of $w[1..|v|+p-1]$ along with $\per(w)$.
By the periodicity lemma, $\per(w)\le \frac12|w|\le |v|$ implies that $\gcd(p-1,\per(w))$ is also a period of that prefix. Thus $\per(w)>p-1$
would contradict the primitivity of $w[1..\per(w)]$.

The algorithm computes the position $p$ using a linear time constant-space pattern matching algorithm. If it exists, it uses letter-by-letter comparison to determine whether $w[1..p-1]$ is a period of $w$. If so, by the discussion above $\per(w)=p-1$ and the algorithm returns this value. 
Otherwise, $2\per(w)> |w|$, i.e., $w$ is not periodic.
The algorithm runs in linear time and uses constant space.

\end{proof}

\subsection{Patterns without Long Highly Periodic Prefix}
\label{sec:aperiodic}
Below we show how to deal with patterns with non-highly periodic prefixes~$\alpha_{j}$.
Patterns with non-highly periodic suffixes $\beta_{j}$ can be processed using the same method after reversing the text
and the patterns. 

\begin{lemma}\label{lem:nhp}
Let $\ell$ be an arbitrary integer.
Suppose we are given a text $T$ of length $n$ and patterns $P_1,\ldots,P_s$ such that for $1 \le j \le s$ we have $\ell \le |P_j| < \frac43 \ell$ and $\alpha_j = P_j[1..\ell]$ is not highly periodic.
We can compute the leftmost and the rightmost occurrence of each pattern $P_j$ in $T$ using $\Oh(n+s(1+\frac{n}{\ell})\log s+s\ell)$ time and $\Oh(s)$ space.
\end{lemma}

The algorithm scans the text $T$ with a sliding window of length $\ell$.
Whenever it encounters a subword equal to the prefix $\alpha_j$ of some $P_j$,
it creates a \emph{request} to verify whether the corresponding suffix $\beta_j$ of length $\ell$
occurs at the appropriate position. The request is processed when the sliding window
reaches that position.
This way the algorithm detects the occurrences of all the patterns. In particular, we may 
store the leftmost and rightmost occurrence of each pattern.

We use the fingerprints to compare the subwords of $T$ with $\alpha_j$ and $\beta_j$.
To this end, we precompute $\Phi(\alpha_j)$ and $\Phi(\beta_j)$ for each $j$.
We also build a deterministic dictionary $\mathcal{D}$~\cite{DBLP:conf/icalp/Ruzic08} with an entry mapping $\Phi(\alpha_j)$ to $j$ for every pattern
(if there are multiple patterns with the same value of $\Phi(\alpha_j)$, the dictionary maps a fingerprint
to a list of indices).
These steps take $\Oh(s\ell)$ and $\Oh(s\log s)$, respectively.
Pending requests are maintained in a priority queue $\mathcal{Q}$,  implemented
using a binary heap\footnote{Hash tables could be used instead of the heap and the deterministic dictionary. Although this would improve the time complexity in Lemma~\ref{lem:nhp}, the running time of the algorithm in Thm.~\ref{thm:final} would not change and
failures with probability inverse polynomial with respect to $s$ would be introduced; see also a discussion before Thm.~\ref{thm:equal}.}
as pairs containing the pattern index (as a value)
and the position where the occurrence of $\beta_j$ is anticipated (as a key).

\DontPrintSemicolon
\begin{algorithm}[hbt]
\For{$i = 1$ \KwSty{to} $n-\ell+1$}{
	$h := \Phi(w[i..i+\ell-1])$\;\label{ln:hash}
	\ForEach{$j : \Phi(\alpha_j)= h$}{\label{ln:dict}
		add a request $(i+|P_j|-\ell, j)$ to $\mathcal{Q}$\label{ln:add}
	}
	\ForEach{request $(i,j)\in \mathcal{Q}$ at position $i$}{\label{ln:front}
		\If{$h = \Phi(\beta_j)$}{\label{ln:check}
		report an occurrence of $P_j$ at $i+\ell-|P_j|$\label{ln:report}
		}
		remove $(i,j)$ from $\mathcal{Q}$\label{ln:remove}
	}
}
\caption{Processing patterns with non-highly periodic $\alpha_j$.}\label{alg:proc}
\end{algorithm}

Algorithm~\ref{alg:proc} provides a detailed description of the processing phase.
Let us analyze its time and space complexities.
Due to the properties of Karp-Rabin fingerprints,
line~\ref{ln:hash} can be implemented in $\Oh(1)$ time. Also, the loops in lines~\ref{ln:dict}
and~\ref{ln:front} takes extra $\Oh(1)$ time even if the respective collections are empty. Apart from these,
every operation can be assigned to a request, each of them taking $\Oh(1)$ (lines~\ref{ln:dict} and~\ref{ln:front}-\ref{ln:check}) or
$\Oh(\log |\mathcal{Q}|)$ (lines~\ref{ln:add} and~\ref{ln:remove}) time.
To bound $|\mathcal{Q}|$, we need to look at the maximum number of pending requests.

\begin{fact}
For any pattern $P_j$ just $\Oh(1+\frac{n}{\ell})$ requests are created
and at any time at most one of them is pending.
\end{fact}

\begin{proof}
Note that there is a one-to-one correspondence between requests concerning $P_j$ 
and the occurrences of $\alpha_j$ in $T$. The distance between two such occurrences
must be at least $\frac{1}{3}\ell$, because otherwise the period of $\alpha_j$ would
be at most $\frac{1}{3}\ell$, thus making $\alpha_j$ highly periodic.
This yields the $\Oh(1+\frac{n}{\ell})$ upper bound on the total number of requests.
Additionally, any request is pending for at most $|P_j|-\ell< \frac{1}{3}\ell$ iterations
of the main \textbf{for} loop. Thus, the request corresponding to an occurrence of $\alpha_j$
is already processed before the next occurrence appears.
\end{proof}

Hence, the scanning phase uses $\Oh(s)$ space and takes $\Oh(n+s(1+\frac{n}{\ell})\log s)$ time.
Taking preprocessing into account, we obtain bounds claimed in Lemma~\ref{lem:nhp}.

\subsection{Highly Periodic Patterns}
\label{sec:periodic}

\begin{lemma}\label{lem:hp}
Let $\ell$ be an arbitrary integer.
Given a text $T$ of length $n$ and a collection of highly periodic patterns $P_1,\ldots,P_s$ such that for $1\le j\le s$ we have
$\ell \le |P_j| < \frac43 \ell$, we can compute
the leftmost occurrence of each pattern $P_{j}$ in $T$ using $\Oh(n+s(1+\frac{n}{\ell})\log s+s\ell)$
total time and $\Oh(s)$ space.
\end{lemma}

The solution is basically the same as in the proof of Lemma~\ref{lem:nhp}, except that the algorithm ignores certain \emph{shiftable} occurrences.
An occurrence of $x$ at position $i$ of $T$ is called \emph{shiftable} if there is another occurrence of $x$ at position $i-\per(x)$.
The remaining occurrences are called \emph{non-shiftable}. Notice that the leftmost occurrence is always non-shiftable, so indeed we can
safely ignore some of the shiftable occurrences of the patterns.
Because $2\per(P_j)\le \frac23 |P_j| \le \frac89\ell < \ell$, the following fact implies that if an occurrence of $P_j$ is non-shiftable, then the occurrence of $\alpha_j$ at the same
position is also non-shiftable.

\begin{fact}
Let $y$ be a prefix of $x$ such that $|y|\ge 2\per(x)$.
Suppose $x$ has a non-shiftable occurrence at position $i$ in $w$.
Then, the occurrence of $y$ at position $i$ is also non-shiftable.
\end{fact}

\begin{proof}
Note that $\per(y)+\per(x)\le |y|$ so the periodicity lemma implies
that $\per(y)=\per(x)$.

Let $x=\rho^k\rho'$ where $\rho$ is the shortest period of $x$.
Suppose that the occurrence of $y$ at position $i$ is shiftable,
meaning that $y$ occurs at position $i-\per(x)$.
Since $|y|\ge \per(x)$, $y$ occurring at position $i-\per(x)$
implies that $\rho$ occurs at the same position.
Thus
$w[i-\per(x)..i+|x|-1]=\rho^{k+1}\rho'$.
But then $x$ clearly occurs at position $i-\per(x)$,
which contradicts the assumption that its occurrence at position $i$ is non-shiftable.
\end{proof}

Consequently, we may generate requests only for the non-shiftable occurrences of $\alpha_j$. In other words,
if an occurrence of $\alpha_{j}$ is shiftable, we do not create the requests and proceed immediately to line~\ref{ln:front}.
To detect and ignore such shiftable occurrences, we maintain the position of the last occurrence of every
$\alpha_{j}$. However, if there are multiple patterns sharing the same prefix
$\alpha_{j_1}=\ldots=\alpha_{j_k}$, we need to be careful so that the time to detect
a shiftable occurrence is $\Oh(1)$ rather than $\Oh(k)$.
To this end, we build another deterministic dictionary, which stores
for each $\Phi(\alpha_j)$ a pointer to the variable where we maintain the position of the previously encountered occurrence
of $\alpha_{j}$. The variable is shared by all patterns with the same prefix $\alpha_{j}$.

It remains to analyze the complexity of the modified algorithm. First, we need to bound the number
of non-shiftable occurrences of a single $\alpha_{j}$. Assume that there is a non-shiftable occurrence
$\alpha_{j}$ at positions $i'<i$ such that $i' \geq i - \frac{1}{2}\ell$. Then $i-i'\leq \frac{1}{2}\ell$ is a period of $T[i'..i+\ell-1]$.
By the periodicity lemma, $\per(\alpha_{j})$ divides $i-i'$, and therefore $\alpha_{j}$ occurs at position
$i'-\per(\alpha_{j})$, which contradicts the assumption that the occurrence at position $i'$ is non-shiftable.
Consequently, the non-shiftable occurrences of every $\alpha_{j}$ are at least $\frac{1}{2}\ell$ characters apart,
and the total number of requests and the maximum number of pending requests can be bounded by
$\Oh(s(1+\frac{n}{\ell}))$ and $\Oh(s)$, respectively, as in the proof of Lemma~\ref{lem:nhp}.
Taking into the account the time and space to maintain the additional components, which
are $\Oh(n+s\log s)$ and $\Oh(s)$, respectively, the final bounds remain the same.

\subsection{Summary}
\begin{theorem}\label{thm:long}
Given a text $T$ of length $n$ and patterns $P_1,\ldots,P_s$ of total length $m$,
using $\Oh(n\log n+m+s\frac{n}{\ell}\log s)$ total time and $\Oh(s)$ space
we can compute
the leftmost occurrences in $T$ of every pattern $P_j$ of length at least $\ell$.
\end{theorem}

\begin{proof}
The algorithm distributes the patterns into $\Oh(\log n)$ groups
according to their lengths, and then into three classes according to their repetitiveness,
which takes $\Oh(m)$ time and $\Oh(s)$ space in total. Then, it applies either Lemma~\ref{lem:nhp}
or Lemma~\ref{lem:hp} on every class. It remains to show that the running times of
all those calls sum up to the claimed bound.
Each of them can be seen as $\Oh(n)$ plus $\Oh(|P_j|+ (1+\frac{n}{|P_j|})\log s)$ per every pattern $P_j$.
Because $\ell \le |P_j| \le n$ and there are $\Oh(\log n)$ groups, this sums up to
$\Oh(n\log n+m+s\frac{n}{\ell}\log s)$.
\end{proof}
Using Thm.~\ref{thm:short} for all patterns of length at most $\min(n,s)$, and (if $s\leq n$)
Thm.~\ref{thm:long} for patterns of length at least $s$, we obtain our main theorem.

\begin{theorem}\label{thm:final}
Given a text $T$ of length $n$ and patterns $P_1,\ldots,P_s$ of total length $m$,
we can compute
the leftmost occurrence in $T$ of every pattern $P_j$ using $\Oh(n\log n
+m)$ total time and $\Oh(s)$ space.
\end{theorem}

\section{Computing Longest Occurring Prefixes}
\label{sec:prefixes}

In this section we extend Thm.~\ref{thm:final} to compute, for every pattern $P_j$,
its longest prefix occurring in the text. A straightforward extension uses binary search
to compute the length $\ell_j$ of the longest prefix of $P_j$ occurring in $T$. All binary
searches are performed in parallel, that is, we proceed in $\Oh(\log n)$ phases.
In every phase we check, for every $j$, if $P_j[1..1+\ell_j]$ occurs in $T$ using
Thm.~\ref{thm:final} and then update the corresponding $\ell_j$ accordingly.
This results in $\Oh(n\log^2 n+m)$ total time complexity. To avoid the logarithmic multiplicative
overhead in the running time, we use a~more complex approach requiring a careful
modification of all the components developed in Sections~\ref{sec:short} and~\ref{sec:long}.

\paragraph{\textbf{Short patterns.}} 

We proceed as in Sect.~\ref{sec:short} while maintaining a tentative longest prefix occurring in $T$
for every pattern $P_j$, denoted $P_j[1..\ell_j]$. 
Recall that after processing a block $T_i$ we obtain, for each pattern $P_j$, a node such that the corresponding substring is equal to $P_j$ provided that $P_j$ occurs in $T_i$. 
Now we need a stronger property, which is that for any length $k$, $\ell_j\le k\le |P_j|$, the ancestor at string depth $k$ of that node (if any) corresponds to $P_j[1..k]$ provided that $P_j[1..k]$ occurs in $T_i$. 
This can be guaranteed by modifying the procedure described in Lemma~\ref{lem:merge}:
if a child of $v_1$ has no corresponding child of $v_2$, we report $v_2$ as the twin of all nodes in the subtree rooted at that child of $v_1$. 
(Notice that now the string depth of $v_1\in \T_1$ might be larger than the string depth of its twin $v_2\in \T_2$, but we generate exactly one twin
for every $v_1\in \T_1$.) 
Using the stronger property we can update every $\ell_j$ by first checking if $P_j[1..\ell_j]$ occurs in $T_i$, and if so incrementing $\ell_j$ as long as possible.
In more detail, let $T_i[b..e]$ denote the substring corresponding to the twin of $P_j$. 
If $|T_i[b..e]|<\ell_j$, there is nothing to do. 
Otherwise, we check whether $T_i[b..(b+\ell_j-1)]=P_j[1..\ell_j]$ using fingerprints, and if so start to naively compare $T_i[b+\ell_j..e]$ and $P_j[\ell_j+1..|P_j|]$. 
Because in the end $\sum_j \ell_j  \leq m$, updating every $\ell_j$ takes $\Oh(s\frac{n}{\ell}+m)$ additional total time.

\begin{theorem}\label{thm:short2}
Given a text $T$ of length $n$ and patterns $P_1,\ldots,P_s$ of total length $m$,
using $\Oh(n\log\ell+s\frac{n}{\ell}+m)$ total time and $\Oh(s+\ell)$ space we can compute
the longest prefix occurring in $T$ for every pattern $P_j$ of length at most $\ell$.
\end{theorem}

\paragraph{\textbf{Long patterns.}}

As in Sect.~\ref{sec:long}, we again distribute all patterns of length at least $\ell$ into groups.
However, now for patterns in the $i$-th group (satisfying $\floor{\log_{4/3}|P_j|}=i$), we set $g_i=\ceil{(4/3)^i}$ and additionally require that $P_j[1..g_i]$ occurs in $T$. 
To verify that this condition is true, we process the groups in the decreasing order of the index $i$ and apply Thm.~\ref{thm:equal} to prefixes $P_j[1..g_i]$. 
If for some pattern $P_j$ the prefix fails to occur in $T$, we replace $P_j$ setting $P_j:= P_j[1..g_i-1]$. 
Observe that this operation moves $P_j$ to a group with a smaller index (or makes $P_j$ a short pattern). 
Additionally, note that in subsequent steps the length of $P_j$ decreases geometrically, so the total length of patterns for which 
we apply Thm.~\ref{thm:equal} is $\Oh(m)$ and thus the total running time of this preprocessing phase is $\Oh(n\log n+m)$
as long as $\ell\ge \log s$. 
%
Hence, from now on we consider only patterns $P_j$ belonging to the $i$-th group, i.e., such that $g_i \leq |P_j|<\frac{4}{3}g_i$ and $P_j[1..g_i]$ occurs in $T$.

As before, we classify patterns depending on their periodicity. 
However, now the situation is more complex, because we cannot reverse the text and the patterns. 
As a warm-up, we first describe how to process patterns $P_j$ with a non-highly periodic prefix $\alpha_j=P_j[1..\ell]$. 
While not used in the final solution, this step allows us to gradually introduce all the required modifications.
 Then we show to process all highly periodic patterns, and finally move to the general case, where patters are not highly periodic.

\paragraph{\textbf{Patterns with a non-highly periodic prefix.}} 
We maintain a tentative longest prefix occurring in $T$ for every pattern $P_j$, denoted $P_j[1..\ell_j]$
and initialized with $\ell_j=\ell$, and proceed as in Algorithm~\ref{alg:proc} with the following
modifications. In line~\ref{ln:add}, the new request is $(i+\ell_j-\ell+1,j)$. In line~\ref{ln:check},
we compare $h$ with $\Phi(P_j[(\ell_j+2-\ell)..(\ell_j+1)])$. If these two fingerprints are equal,
we have found an occurrence of $P_j[1..\ell_j+1]$. In such case we try to further extend the occurrence
by naively comparing $P_j[\ell_j+1..|P_j|]$ with the corresponding fragment of $T$
and incrementing $\ell_j$ as long as the corresponding characters match.
For every $P_j$ we also need to maintain $\Phi(P_j[(\ell_j+2-\ell)..(\ell_j+1)])$, which can be 
first initialized in $\Oh(\ell)$ time and then updated in $\Oh(1)$ time whenever $\ell_j$ is incremented.
Because at any time at most one request is pending for every pattern $P_j$ (and thus, while updating $\ell_j$ no such request is pending), this modified algorithm correctly determines the longest occurring prefix for every pattern with non-highly periodic $\alpha_j$.

\begin{lemma}\label{lem:nhp3}
Let $\ell$ be an arbitrary integer.
Suppose we are given a text $T$ of length $n$ and patterns $P_1,\ldots,P_s$ such that, for \(1 \leq j \leq s\), we have $\ell \le |P_j| < \frac43 \ell$ and $\alpha_j = P_j[1..\ell]$ is not highly periodic.
We can compute the longest prefix occurring in $T$ for every pattern $P_j$
using $\Oh(n+s(1+\frac{n}{\ell})\log s+s\ell)$ total time using $\Oh(s)$ space.
\end{lemma}

\paragraph{\textbf{Highly periodic patterns.}}

As in Sect.~\ref{sec:periodic}, we observe that all shiftable occurrences of the longest prefix of
$P_j$ occurring in $T$ can be ignored, and therefore it is enough to consider only non-shiftable occurrences
of $\alpha_j$ (by the same argument, because that longest prefix is of length at least $|\alpha_j|$).
Therefore, we can again use Algorithm~\ref{alg:proc} with the same modifications.
As for non-highly periodic $\alpha_j$, we maintain a tentative longest prefix $P_j[1..\ell_j]$
for every pattern $P_j$. Whenever a non-shiftable occurrence of $\alpha_j$ is detected, 
we create a new request to check if $\ell_j$ can be incremented. If so, we start to naively
compare $P_j[1..\ell_j+1]$ with the corresponding fragment of $T$. The total time and space
complexity remain unchanged.

\begin{lemma}\label{lem:hp2}
Let $\ell$ be an arbitrary integer.
Given a text $T$ of length $n$ and a collection of highly periodic patterns $P_1,\ldots,P_s$ such that,
for $1\leq j\leq s$, we have
$\ell \le |P_j| < \frac43 \ell$, we can compute the longest prefix occurring in $T$ for
every pattern $P_j$ using $\Oh(n+s(1+\frac{n}{\ell})\log s+s\ell)$ total time and $\Oh(s)$ space.
\end{lemma}

\paragraph{\textbf{General case.}}

Now we describe how to process all non-highly periodic patterns $P_j$. This will be an extension
of the simple modification described for the case of non-highly periodic prefix $\alpha_j$. 
We start with the following simple combinatorial fact.

\begin{fact}
\label{fct:break}
Let $\ell \geq 3$ be an integer and $w$ be a non-highly periodic word of length at least $\ell$.
Then there exists $i$ such that $w[i..i+\ell-1]$ is not highly periodic and either $i=1$ or $w[1..i+\ell-2]$
is highly periodic. Furthermore, such $i$ can be found in $\Oh(|w|)$ time and constant space
assuming read-only random access to $w$.
\end{fact}

\begin{proof}
If $\per(w[1..\ell]) > \frac{1}{3}\ell$, we are done. Otherwise, choose largest $j$ such that
$\per(w[1..j]) = \per(w[1..\ell])$. Since $w$ is not highly periodic, we have $j<|w|$. 
Thus $\per(w[1..j]) \leq \frac{1}{3}\ell$ but $\per(w[1..j+1]) > \frac{1}{3}\ell$.
We claim that $i=j+2-\ell$ can be returned. We must argue that $\per(w[j+2-\ell..j+1]) > \frac{1}{3}\ell$.
Otherwise, the periods of both $w[1..j]$ and $w[j+2-\ell..j+1]$ are at most $\frac{1}{3}\ell$.
But these two substrings share a fragment of length $\ell-1 \geq \frac{2}{3}\ell$, so by the periodicity
lemma their periods are in fact the same, and then the whole $w[1..j+1]$ has period at most $\frac{1}{3}\ell$,
which is a contradiction.

Regarding the implementation, we compute $\per(w[1..\ell])$ using Lemma~\ref{lem:comper}.
Then we check how far the period of $w[1..\ell]$ extends in the whole $w$ naively in $\Oh(|w|)$
time and constant space.
\end{proof}

For every non-highly periodic pattern $P_j$ we use Fact~\ref{fct:break} to find its non-highly
periodic substring of length $\ell$, denoted $P_j[k_j...k_j+\ell-1]$, such that $k_j=1$ or
$P_j[1..k_j+\ell-2]$ is highly periodic. We begin with checking if $P_j[1..k_j+\ell-1]$
occurs in $T$ using Lemma~\ref{lem:nhp} (if $k_j=1$, it surely does because of how we partition
the patterns into groups).
If not, we replace $P_j$ setting $P_j:= P_j[1..k_j+\ell-2]$, which is highly periodic and can be processed as already described. 
From now on we consider only patterns $P_j$ such that $P_j[k_j..k_j+\ell-1]$
is not highly periodic and $P_j[1..k_j+\ell-1]$ occurs in $T$.

\newcommand{\dummy}{P_j[k_j..k_j+\ell\!-\!1]}

\DontPrintSemicolon
\begin{algorithm}[hbt]
 \For{$j = 1$ \KwSty{to} $s$}{
   $\ell_j := k_j+\ell-1$\;
 }
 \For{$i = 1-\frac{1}{3}\ell$ \KwSty{to} $n-\ell+1$}{
 	$h_1 := \Phi(w[i+\frac{1}{3}\ell..i+\frac{4}{3}\ell-1])$\;
	$h_2 := \Phi(w[i+\ell..i+\ell-1])$\;
	\ForEach{$j : \Phi(P_j[k_j..k_j+\ell-1])= h_1$}{
		add a request $(i+\frac{1}{3}\ell-k_j+1, j)$ to $\mathcal{Q}_1$
	}
	\ForEach{request $(i,j)\in \mathcal{Q}_1$ at position $i$}{
		\If{$h_2 = \Phi(P[1..\ell])$}{
  		  add a request $(i-\ell+\ell_j+1, j)$ to $\mathcal{Q}_2$\;
		}
		remove $(i,j)$ from $\mathcal{Q}_1$\;
	}
	\ForEach{request $(i-\ell,j)\in \mathcal{Q}_2$ at position $i$}{
		\If{$h_2 = \Phi(P_j[\ell_j+2-\ell..\ell_j+1])$}{
                    $o_j := i+\ell -\ell_j-1$\;
                    increment $\ell_j$ as long as $P_j[\ell_j+1]=T[o_j+\ell_j]$\;
		}
		remove $(i,j)$ from $\mathcal{Q}_2$\;
	}

 }
\caption{Processing patterns with non-highly periodic $\dummy$.
}\label{alg:proc2}
\end{algorithm}

We further modify Algorithm~\ref{alg:proc} to obtain Algorithm~\ref{alg:proc2} as follows. We scan
the text $T$ with a sliding window of length $\ell$ while maintaining a tentative longest occurring
prefix $P_j[1..\ell_j]$ for every pattern $P_j$, initialized by setting $\ell_j=k_j+\ell-1$.
Whenever we encounter a substring equal to
$P_j[k_j..k_j+\ell-1]$, i.e., $P_j[k_j..k_j+\ell-1]=T[i..i+\ell-1]$, we want to check if
$P_j[1..k_j+\ell-1]=T[i-k_j+1..i+\ell-1]$ by comparing fingerprints
of $\alpha_j=P_j[1..\ell]$ and the corresponding fragment of $T$. This is not trivial as that fragment
is already to the left of the current window. Hence we conceptually move \emph{two} sliding windows of
length $\ell$, corresponding to $w[i+\frac{1}{3}\ell..i+\frac{4}{3}\ell-1]$ and $w[i..i+\ell-1]$, respectively.
Because $k_j \leq |P_j|-\ell \le \frac{1}{3}\ell$, whenever the first window generates a request
(called request of type I), the second one is still far enough to the left for the request to 
be processed in the future. Furthermore, because $P_j[k_j..k_j+\ell-1]$ is non-highly periodic and the
distance between the sliding windows is $\frac{1}{3}\ell$, each pattern $P_j$ contributes at most one pending request of type I
at any moment and $\Oh(1+\frac{n}{\ell})$ such requests in total. 
Then, whenever a request of type I is successfully processed, we know that $P_j[1..k_j+\ell-1]$ matches with the corresponding fragment of $T$.
We want to check if the occurrence of $P_j[1..k_j+\ell-1]$ can be extended to an occurrence
of $P_j[1..\ell_j+1]$. To this end, we create another request (called request of type II) to check if $P_j[\ell_j+2-\ell..\ell_j+1]$ matches the corresponding fragment of $T$. 
This request can be processed using the second window and, again because $\per(P_j[1..k_j+\ell-1])>\frac{1}{3}\ell$, each pattern contribues at most one pending request of type II at any moment 
and $\Oh((1+\frac{n}{\ell}))$ such request in total.
Finally, whenever a request of type II
is successfully processed, we know that the corresponding $\ell_j$ can be incremented.
Therefore, we start to naively compare the characters of $P_j[\ell_j+1..|P_j]]$ and the corresponding fragment
of $T$. 
Since at that time no other request of type II is pending for $P_j$, such modified algorithm correctly computes all values $\ell_j$ (and the corresponding positions $o_j$).

\begin{lemma}\label{lem:nhp2}
Let $\ell$ be an arbitrary integer.
Suppose we are given a text $T$ of length $n$ and patterns $P_1,\ldots,P_s$ such that, for \(1 \leq j \leq s\), we have $\ell \le |P_j| < \frac43 \ell$ and $P_j$ is non-highly periodic.
We can compute the longest prefix occurring in $T$ for every pattern $P_j$
in $\Oh(n+s(1+\frac{n}{\ell})\log s+s\ell)$ total time using $\Oh(s)$ space.
\end{lemma}

By combining all the ingredients, we get the following theorem.

\begin{theorem}\label{thm:final2}
Given a text $T$ of length $n$ and patterns $P_1,\ldots,P_s$ of total length $m$,
we can compute
the longest prefix occurring in $T$ for every pattern $P_j$ using $\Oh(n\log n
+m)$ total time and $\Oh(s)$ space.
\end{theorem}

\begin{proof}
We proceed as in Thm.~\ref{thm:long} and~\ref{thm:final}, except that now we use
Thm.~\ref{thm:short2}, Lemma~\ref{lem:hp2} and Lemma~\ref{lem:nhp2} instead of
Thm.~\ref{thm:short}, Lemma~\ref{lem:hp} and Lemma~\ref{lem:nhp}, respectively.
Additionally, we need $\Oh(n\log n+m)$ time to distribute the long patterns
into groups, which is absorbed in the final complexity.
\end{proof}

Finally, let us note that it is straightforward to modify the algorithm so that we can specify
for every pattern $P_j$ an upper bound $r_j$ on the starting positions of the occurrences.

\begin{theorem}\label{thm:gen}
Given a text $T$ of length $n$, patterns $P_1,\ldots,P_s$ of total length $m$, and integers $r_1,\ldots,r_s$,
we can compute for each pattern the maximum length $\ell_j$ and a position $o_j\le r_j$ such that $P_j[1..\ell_j]=T[o_j..(o_j+\ell_j-1)]$,
using $\Oh(n\log n+m)$ total time and $\Oh(s)$ space.
\end{theorem}

\section{Las Vegas Algorithms}\label{app:lv}

As shown below, it is not difficult to modify our dictionary matching algorithm so that it 
always verifies the correctness of the answers.
Assuming that we are interested in finding just the leftmost occurrence of every pattern,
we obtain an $\Oh(n\log n+m)$-time Las Vegas algorithm (with inverse-polynomial failure probability).

In most cases, it suffices to naively verify in $\Oh(|P_j|)$ time whether the leftmost occurrence of $P_j$ detected by the algorithm is valid. 
If it is not, we are guaranteed that the fingerprints $\Phi$ admit a false-positive. Since this event happens
with inverse-polynomial probability, a failure can be reported.

This simple solution remains valid for short patterns and non-highly periodic long patterns. 
For highly periodic patterns, the situation is more complicated. 
The algorithm from Lemma~\ref{lem:hp} assumes that we are correctly detecting all
occurrences of every $\alpha_{j}$ so that we can filter out the shiftable ones.
Verifying these occurrences naively might take
too much time, because it is not enough check just one
occurrence of every~$\alpha_{j}$.

Recall that $\per(\alpha_{j})\leq \frac{1}{3}\ell$.
If the previous occurrence of $\alpha_{j}$ was at position $i\geq i-\frac{1}{2}\ell$,
we will check if $\per(\alpha_{j})$ is a period of $T[i'..i+\ell-1]$.
If so, either both occurrences (at position $i'$ and at position $i$)
are false-positives, or none of them is, and the occurrence at position $i'$
can be ignored. Otherwise, at least one occurrence is surely false-positive, and we declare a failure. 
To check if $\per(\alpha_{j})$ is a period of $T[i'..i+\ell-1]$,
we partition $T$ into overlapping blocks $T_{1}=[1..\frac{2}{3}\ell], T_{2}=[\frac{1}{3}\ell+1..\frac{4}{3}\ell], \ldots$.
Let $T_{t}=T[(t-1)\frac{1}{3}\ell+1..(t+1)\frac{1}{3}\ell]$ be the rightmost such block fully inside $T[1..i+\ell-1]$.
We calculate the period of $T_{t}$ using Lemma~\ref{lem:comper} in $\Oh(\ell)$
time and $\Oh(1)$ space, and then calculate  how far the period extends
to the left and to the right, terminating if it extends very far.
Formally, we calculate the largest $e< (t+2)\frac{1}{3}\ell$ and the smallest
$b > (t-2)\frac{1}{3}\ell-\frac{1}{2}\ell$ such that $\per(T_{t})$ is a period
of $T[b..e]$. This takes $\Oh(\ell)$ time for every $t$ summing up to $\Oh(n)$ total time.
Then, to check if $\per(\alpha_{j})$ is a period of
$T[i'..i+\ell-1]$ we check if it divides the period of $\per(T_{t})$ and furthermore
$r \geq i+\ell-1$ and $\ell \leq i'$. 
Finally, we naively verify the the reported leftmost occurrences of $P_j$.
Consequently, Las Vegas randomization suffices in Thoerem~\ref{thm:final}.

For Thm.~\ref{thm:final2}, we run the Las Vegas version of Thm.~\ref{thm:final} with $P_j[1..\ell_j]$ and $P_j[1..\ell_j+1]$ as patterns to make sure that the former occur in $T$ but the latter do not. 
For Thm.~\ref{thm:gen} we also use Thm.~\ref{thm:final}, but this time we need to see where the reported leftmost occurrences start compared to bounds $r_j$.

\section{Approximating LZ77 in Small Space}
\label{sec:lz77approx}
A non-empty fragment $T[i..j]$ is called a \emph{previous fragment} if the corresponding subword
occurs in $T$ at a position $i'<i$.
A \emph{phrase} is either a previous fragment or a single letter not occurring before in $T$.
The LZ77-factorization of a text $T[1..n]$ is a greedy factorization of $T$ into $z$ \emph{phrases},
$T=f_1f_2\dots f_z$, such that each $f_i$ is as long as possible.
To formalize the concept of LZ77-approximation, we first make the following definition.

\begin{definition}
Let $w=g_1g_2\dots g_a$ be a factorization of $w$ into $a$ phrases. We call it
$c$-optimal if the fragment corresponding to the concatenation of any $c$ consecutive phrases $g_i\ldots g_{i+c-1}$
is not a previous fragment.
\end{definition}

A $c$-optimal factorization approximates the LZ77-factorization in the number of factors, as the following observation
states.
However, the stronger property of $c$-optimality is itself useful in certain situations.

\begin{observation}\label{obs:approx}
If $w=g_{1}g_{2}\dots g_{a}$ is a $c$-optimal factorization of $w$ into $a$ phrases, and the LZ77-factorization of $w$ consists of
$z$ phrases, then $a\leq  c \cdot z$.
\end{observation}

We first describe how to use the dictionary matching algorithm described in Thm.~\ref{thm:final}
to produce a $2$-optimal factorization of $w[1..n]$ in $\Oh(n\log n)$ time and $\Oh(z)$ working
space. Then, the resulting parse can be further refined to produce a $(1+\eps)$-optimal factorization
in $\Oh(\eps^{-1}n\log n)$ additional time and the same space using the extension of the dictionary matching
algorithm from Thm~\ref{thm:final2}.

\subsection{2-Approximation Algorithm}

\subsubsection{Outline.}
Our algorithm is divided into three phases, each of which refines the factorization from the previous phase:
\begin{enumerate}[\bfseries{Phase }1.]
	\item Create a factorization of $T[1..n]$ stored implicitly as $z$ \emph{chains}
	consisting of $\Oh(\log n)$ phrases each.
	\item Try to merge phrases within the chains to produce an $\Oh(1)$-optimal factorization.
	\item Try to merge adjacent factors as long as possible to produce the final $2$-optimal factorization.
\end{enumerate}
Every phase takes $\Oh(n\log n)$ time and uses $\Oh(z)$ working space. In the end, we get a 2-approximation of the LZ77-factorization.
Phases 1 and 2 use the very simple multiple pattern matching algorithm for patterns of equal lengths developed
in Thm.~\ref{thm:equal}, while Phase 3 requires the general multiple pattern matching algorithm obtained in Thm.~\ref{thm:final}.

\subsubsection{Phase 1.}
To construct the factorization, we imagine creating a binary tree on top the text $T$ of length $n=2^k$ -- see also Fig.~\ref{fig:example_approx} (we implicitly pad $w$ with
sufficiently many $ \$$'s to make its length a power of $2$).
The algorithm works in $\log n$ rounds, and the $i$-th round works on level $i$ of the tree, starting at $i=1$ (the children of the root).
On level $i$, the tree divides $T$ into $2^i$ blocks of size $n/2^i$; the aim is to identify previous fragments
among these blocks and declare them as phrases.
(In the beginning, no phrases exist, so all blocks are unfactored.)
To find out if a block is a previous fragment, we use Thm.~\ref{thm:equal} and test whether the leftmost occurrence of the corresponding subword is the block itself.
The exploration of the tree is naturally terminated at the nodes corresponding to the previous fragments (or single letters not occurring before), forming the leaves of a (conceptual) binary tree.
A pair of leaves sharing the same parent is called a \emph{cherry}. The block corresponding to the common parent
is \emph{induced} by the cherry.
To analyze the algorithm, we make the following observation:

\begin{figure}[t]
	\begin{center}
		\includegraphics[scale=1.2]{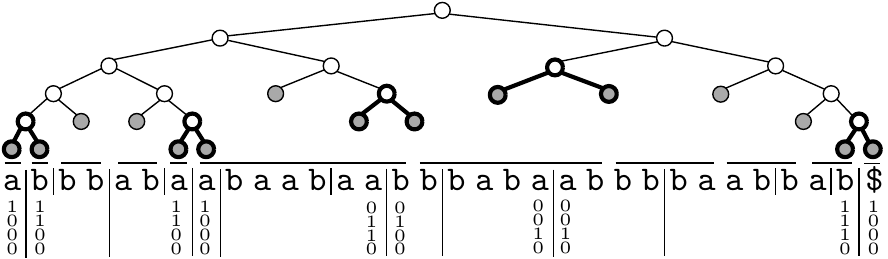}
	\end{center}
	\caption{An illustration of Phase 1 of the algorithm, with the ``cherries'' depicted in thicker lines. The horizontal lines represent the LZ77-factorization and the vertical lines depict factors induced by the tree. Longer separators are drawn between chains, whose lengths are written in binary with the least significant bits on top.
	}
		\label{fig:example_approx}
\end{figure}

\begin{fact}
\label{fct:cherries}
A block induced by a cherry is never a previous fragment.
Therefore, the number of cherries is at most $z$.
\end{fact}
\begin{proof}
The former part follows from construction. To prove the latter,
observe that the blocks induced by different cherries are disjoint and hence each cherry can be assigned a unique LZ77-factor
ending within the block.
\end{proof}

Consequently, while processing level $i$ of the tree, we can afford storing all cherries generated so far on a sorted linked list $\llist$. The remaining already
generated phrases are not explicitly stored.
In addition, we also store a sorted linked list $\llist_i$ of all still \emph{unfactored} nodes on the current level $i$ (those for which the corresponding
blocks are tested as previous fragments).
Their number is bounded by $z$ (because there is a cherry below every node on the list),
so the total space is $\Oh(z)$.
Maintaining both lists sorted is easily accomplished by scanning them
in parallel with each scan of $T$, and inserting new cherries/unfactored nodes at their correct places.
Furthermore, in the $i$-th round we apply Thm.~\ref{thm:equal} to at most $2^{i}$ patterns of length $n/2^{i}$, so the total time is
$\sum_{i=1}^{\log n}\Oh(n+2^{i}\log(2^{i}))=\Oh(n\log n)$.

Next, we analyze the structure of the resulting factorization. Let $h_{x-1}h_x$ and $h_{y}h_{y+1}$ be the two consecutive
cherries. The phrases $h_{x+1}\ldots h_{y-1}$ correspond to the right siblings
of the ancestors of $h_x$ and to the left siblings of the ancestors of $h_y$ (no further than to the lowest common ancestor of $h_x$ and $h_y$).
This naturally partitions $h_x h_{x+1}\ldots h_{y-1}h_y$ into two parts,
called an \emph{increasing chain} and a \emph{decreasing chain} to depict the behaviour of phrase
lengths within each part. Observe that these lengths are powers of two, so the structure of a chain
of either type is determined by the total length of its phrases, which can be interpreted
as a bitvector with bit $i'$ set to 1 if there is a phrase of length $2^{i'}$ in the chain. 
Those bitvectors can be created while traversing the tree level by level, passing the partially 
created bitvectors down to the next level $\llist_{i+1}$ until finally storing them at the cherries in $\llist$. 

At the end we obtain a sequence of chains of alternating types,
see Fig.~\ref{fig:example_approx}.
Since the structure of each chain follows from its length, we store the sequence of chains
rather the actual factorization, which might consist of $\Theta(z\log n)=\omega(z)$ phrases.
By Fact~\ref{fct:cherries}, our representation uses $\Oh(z)$ words of space and
the last phrase of a decreasing chain
concatenated with the first phrase of the consecutive increasing chain never form a previous fragment (these phrases form the block
induced by the cherry).

\subsubsection{Phase 2.}
In this phase we merge phrases within the chains. 
We describe how to process increasing chains; the decreasing are handled, mutatis mutandis, analogously.
We partition the phrases $h_{\ell}\ldots h_r$ within a chain into groups.

For each chain we maintain an \emph{active} group, initially consisting of $h_\ell$, and scan the remaining phrases in the left-to-right order.
We either append a phrase $h_i$ to the active group $g_j$,
or we output $g_j$ and make $g_{j+1}=h_i$ the new active group.
The former action is performed if and only if the fragment of length $2|h_i|$ starting at the same position as $g_j$ is a previous fragment.
Having processed the whole chain, we also output the last active group.

\begin{fact}
\label{fct:nothree}
Within every chain every group $g_j$ forms a valid phrase,
but no concatenation of three adjacent groups $g_jg_{j+1}g_{j+2}$ form a previous fragment.
\end{fact}

\begin{proof}
Since the lengths of phrases form an increasing sequence of powers of two, at the moment we need to decide
if we append $h_i$ to $g_j$ we have $|g_j|\le |h_\ell\ldots h_{i-1}|<|h_i|$,
so $2|h_i|>|g_jh_i|$, and thus we are guaranteed if we append $g_j$,  then $g_jh_i$ is a previous factor.
Finally, let us prove the aforementioned optimality condition, i.e., that $g_jg_{j+1}g_{j+2}$ is not a previous fragment for any three consecutive groups.
Suppose that we output $g_j$ while processing $h_i$, that is, $g_{j+1}=h_{i}\ldots h_{i'}$.
We did not append $h_i$ to $g_j$, so the fragment of length $2|h_i|$ starting at the same position as $g_j$ is not a previous fragment.
However, $|g_jg_{j+1}g_{j+2}|>|g_{j+1}g_{j+2}|\ge |h_{i}h_{i+1}|>2|h_i|$, so this immediately implies that $g_jg_{j+1}g_{j+2}$ is not a previous fragment. 
\end{proof}

The procedure described above is executed in parallel for all chains, each of which maintains just the length of its active group.
In the $i$-th round only chains containing a phrase of length $2^i$ participate (we use bit operations to verify which
chains have length containing $2^{i}$ in the binary expansion).
These chains provide fragments of length $2^{i+1}$
and Thm.~\ref{thm:equal} is applied to decide
which of them are previous fragments. The chains modify their active groups based on the answers;
some of them may output their old active groups.
These groups form phrases of the output factorization, so the space required to store them is
amortized by the size of this factorization. As far as the running time is concerned,
we observe that no more than $\min(z,\frac{n}{2^i})$ chains participate in the $i$-th round .
Thus, the total running time is
$\sum_{i=1}^{\log n}\Oh(n+\frac{n}{2^{i}}\log\frac{n}{2^{i}})=\Oh(n\log n)$.
To bound the overall approximation guarantee, suppose there are five consecutive output phrases forming a previous fragment.
By Fact~\ref{fct:cherries}, these fragments cannot contain a block induced by any cherry. Thus, the phrases
are contained within two chains. However, by Fact~\ref{fct:nothree} no three consecutive phrases obtained
from a single chain form a previous fragment. Hence the resulting factorization is 5-optimal.

\subsubsection{Phase 3.}
The following lemma achieves the final 2-approximation:
\begin{lemma}
Given a $c$-optimal factorization, one can compute a 2-optimal factorization
using $\Oh(c\cdot n\log n)$ time and $\Oh(c\cdot z)$ space.
\end{lemma}
\begin{proof}
The procedure consists of $c$ iterations. In every iteration we first
detect previous fragments corresponding to concatenations of two adjacent phrases.
The total length of the patterns is up to $2n$, so this takes $\Oh(n\log n + m)=\Oh(n\log n)$ time and $\Oh(c\cdot z)$
space using Thm.~\ref{thm:final}.
Next, we scan through the factorization and merge every phrase $g_i$  with the preceding phrase $g_{i-1}$ 
if $g_{i-1}g_i$ is a previous fragment and $g_{i-1}$ has not been just merged with its predecessor.

We shall prove that the resulting factorization is 2-optimal. Consider a pair of adjacent phrases $g_{i-1}g_i$ in the final
factorization
and let $j$ be the starting position of $g_i$. Suppose $g_{i-1}g_i$ is a previous fragment. Our algorithm performs
merges only, so the phrase ending at position $j-1$ concatenated with the phrase starting at position $j$ 
formed a previous fragment at every iteration. The only reason that these factors were not merged could
be another merge of the former factor. Consequently, the factor ending at position $j-1$ took part in a merge
at every iteration, i.e., $g_{i-1}$ is a concatenation of at least $c$ phrases of the input factorization.
However, all the phrases created by the algorithm form previous fragments, which contradicts the $c$-optimality
of the input factorization.
\end{proof}

\subsection{Approximation Scheme}

The starting point is a 2-optimal factorization into $a$ phrases, which can be found in $\Oh(n\log n)$ time
using the previous method. The text is partitioned into $\frac{\eps}{2}a$ blocks corresponding to 
$\frac{2}{\eps}$ consecutive phrases. Every block is then greedily factorized into phrases.
The factorization is implemented using Thm.~\ref{thm:gen} to compute the longest previous fragments
in parallel for all blocks as follows. Denote the starting positions of the blocks by $b_1<b_2<\ldots$
and let $i_j$ be the current position in the $i$-th block, initially set to $b_i$.
For every $i$, we find the longest previous fragment starting at $i_j$ and fully contained inside
$T[i_j..b_{j+1}-1]$ by computing the longest prefix of $T[i_j..b_{j+1}-1]$ occurring in $T$
and starting in $T[1..i_j-1]$, denoted $T[i_j..i_j+\ell_j-1]$. Then we output every
$T[i_j..i_j+\ell_j-1]$ as a new phrase and increase $i_j$ by $\ell_j$. Because every block,
by definition, can be factorized into $\frac{2}{\eps}$ phrases and the greedy factorization is optimal,
this requires $\Oh(\frac{1}{\eps}n\log n)$ time in total.
To bound the approximation guarantee, observe that
every phrase inside the block, except possibly for the last one, contains an endpoint of a phrase in the
LZ77-factorization. Consequently, the total number of phrases is at most $z+\frac{\eps}{2}a\le(1+\eps)z$.

\begin{theorem}\label{thm:finalLZ}
Given a text $T$ of length $n$ whose LZ77-factorization consists of $z$ phrases, we can factorize
$T$ into at most $2z$ phrases using $\Oh(n\log n)$ time and $\Oh(z)$ space.
Moreover, for any $\eps\in (0,1]$ in $\Oh(\eps^{-1} n\log n)$ time and $\Oh(z)$ space we can compute
a factorization into no more than $(1+\eps)z$ phrases.
\end{theorem}

\paragraph{\textbf{Acknowledgments.}}
The authors would like to thank the participants of the Stringmasters 2015 workshop in Warsaw,
where this work was initiated. We are particularly grateful to Marius Dumitran, Artur Jeż, and Patrick K.\ Nicholson.

\bibliographystyle{abbrv}
\bibliography{mult}

\end{document}